\documentclass[a4paper,UKenglish]{lipics-v2016}
 
\usepackage{microtype}

\title{A Normalizing Computation Rule for Propositional Extensionality in Higher-Order Minimal Logic
}
\titlerunning{Propositional Extensionality in Higher-Order Minimal Logic} 

\author[1]{Robin Adams}
\author[1]{Marc Bezem}
\author[2]{Thierry Coquand}
\affil[1]{Universitetet i Bergen, Institutt for Informatikk, Postboks 7800, N-5020 BERGEN, Norway \\
  \texttt{\{robin.adams,bezem\}@ii.uib.no}}
\affil[2]{Chalmers tekniska högskola, Data- och informationsteknik, 412 96 Göteborg, Sweden \\
  \texttt{coquand@chalmers.se}}
\authorrunning{R. Adams, M. Bezem and T. Coquand} 

\Copyright{Robin Adams and Marc Bezem and Thierry Coquand}

\subjclass{Dummy classification -- please refer to \url{http://www.acm.org/about/class/ccs98-html}}

\EventEditors{John Q. Open and Joan R. Acces}
\EventNoEds{2}
\EventLongTitle{42nd Conference on Very Important Topics (CVIT 2016)}
\EventShortTitle{CVIT 2016}
\EventAcronym{CVIT}
\EventYear{2016}
\EventDate{December 24--27, 2016}
\EventLocation{Little Whinging, United Kingdom}
\EventLogo{}
\SeriesVolume{42}
\ArticleNo{23}

\usepackage[all,cmtip]{xy}
\usepackage{suffix}

\newcommand*{\eqdef}{\mathrel{\smash{\stackrel{\text{def}}{=}}}}
\newcommand*{\isotoid}{\ensuremath{\mathsf{isotoid}}}
\newcommand*{\reff}[1]{\ensuremath{\mathrm{ref} \left( {#1} \right)}}
\newcommand*{\univ}[4]{\ensuremath{\mathrm{univ}_{{#1}, {#2}} \left({#3} , {#4} \right)}}
\WithSuffix\newcommand\univ*[2]{\ensuremath{\mathsf{univ} \left( {#1} , {#2} \right)}}
\newcommand*{\triplelambda}{\ensuremath{\lambda \!\! \lambda \!\! \lambda}}
\newcommand*{\vald}{\ensuremath{\vdash \mathrm{valid}}}
\newcommand*{\dom}{\ensuremath{\operatorname{dom}}}
\newcommand{\Path}[3]{\ensuremath{\mathsf{Path} \, {#1} \, {#2} \, {#3}}}
\newcommand{\Prop}{\mathsf{Prop}}
\newcommand{\outputt}{\mathsf{output}}
\newcommand{\isProp}[1]{\mathsf{isProp} \left( {#1} \right)}
\newcommand{\mapid}[2]{\mathsf{mapid}_{#1} \, {#2}}
\newcommand{\comp}{\mathsf{comp}}
\newcommand{\id}{\mathsf{id}}

\theoremstyle{plain}
\newtheorem{proposition}[theorem]{Proposition}
\theoremstyle{definition}
\newtheorem{note}[theorem]{Note}

\usepackage{proof}

\begin{document}

\maketitle

\begin{abstract}
The univalence axiom expresses the principle of extensionality for dependent type theory. However, if we simply add the univalence axiom to type theory, then we lose the property of \emph{canonicity} --- that every closed term computes to a canonical form. A computation becomes `stuck' when it reaches the point that it needs to evaluate a proof term that is an application of the univalence axiom. So we wish to find a way to compute with the univalence axiom. While this problem has been solved with the formulation of cubical type theory, where the computations are expressed using a nominal extension of lambda-calculus, it may be interesting to explore alternative solutions, which do not require such an extension.

As a first step, we present here a system of propositional higher-order minimal logic (PHOML).  There are three kinds of typing judgement in PHOML.  There are \emph{terms} which inhabit \emph{types}, which are the simple types over $\Omega$.  There are \emph{proofs} which inhabit \emph{propositions}, which are the terms of type $\Omega$.  The canonical propositions are those constructed from $\bot$ by implication $\supset$.  Thirdly, there are \emph{paths} which inhabit \emph{equations} $M =_A N$, where $M$ and $N$ are terms of type $A$.  There are two ways to prove an equality: reflexivity, and \emph{propositional extensionality} --- logically equivalent propositions are equal.  This system allows for some definitional equalities that are not present in cubical type theory, namely that transport along the trivial path is identity.

We present a call-by-name reduction relation for this system, and prove that the system satisfies canonicity: every closed typable term head-reduces to a canonical form.  This work has been formalised in Agda.
 \end{abstract}


\section{Introduction}

The \emph{univalence axiom} of Homotopy Type theory (HoTT) \cite{hottbook} postulates a
constant
$$ \isotoid : A \simeq B \rightarrow A = B $$
that is an inverse to the obvious function $A = B \rightarrow A \simeq B$.  However, if we simply add this constant to Martin-L\"{o}f type theory, then
we lose the important property of \emph{canonicity} --- that every closed term of type $A$ computes to a unique canonical object of type $A$.  When a computation reaches a point
where we eliminate a path (proof of equality) formed by $\isotoid$, it gets 'stuck'.

As possible solutions to this problem, we may try to do with a weaker property than canonicity, such as \emph{propositional canonicity}:
that every closed term of type $\mathbb{N}$ is \emph{propositionally} equal to a numeral, as conjectured by Voevodsky.  Or we may attempt to change the definition of equality to make $\isotoid$ definable \cite{Polonsky14a}, or add a nominal extension to the syntax of the type theory (e.g. Cubical Type Theory \cite{cchm:cubical}).

We could also try a more conservative approach, and simply attempt to find a reduction relation for a type theory involving $\isotoid$ that satisfies
all three of the properties above.  There seems to be no reason \emph{a priori} to believe this is not possible, but it is difficult to do because
the full Homotopy Type Theory is a complex and interdependent system.  We can tackle the problem by adding univalence to a much simpler system, finding
a well-behaved reduction relation, then doing the same for more and more complex systems, gradually approaching the full strength of HoTT.

In this paper, we present a system we call PHOML, or predicative higher-order minimal logic.  It is a type theory with three kinds of typing judgement.  There are \emph{terms} which inhabit \emph{types}, which are the simple types over $\Omega$.  There are \emph{proofs} which inhabit \emph{propositions}, which are the terms of type $\Omega$.  The canonical propositions are those constructed from $\bot$ by implication $\supset$.  Thirdly, there are \emph{paths} which inhabit \emph{equations} $M =_A N$, where $M$ and $N$ are terms of type $A$.

There are two canonical forms for proofs of $M =_\Omega N$.  For any term $\varphi : \Omega$, we have $\reff{\varphi} : \varphi =_\Omega \varphi$.  We also add univalence for this system, in this form:
if $\delta : \varphi \supset \psi$ and $\epsilon : \psi \supset\varphi$, then $\univ{\varphi}{\psi}{\delta}{\epsilon} : \varphi =_\Omega \psi$.  

This entails that in PHOML, two propositions that are logically equivalent are equal.  Every function of type $\Omega \rightarrow \Omega$ that can be constructed in PHOML must therefore respect logical equivalence.  That is, for any $F$ and logically equivalent $x$, $y$ we must have that $Fx$ and $Fy$ are logically equivalent. Moreover, if for $x:\Omega$ we have that $Fx$ is logically equivalent to $Gx$, then $F =_{\Omega\to\Omega} G$.
Every function of type $(\Omega \rightarrow \Omega) \rightarrow \Omega$ must respect this equality; and so on.  This is the manifestation in PHOML of the principle that only homotopy invariant constructions can be performed in homotopy type theory.  (See Section \ref{section:exampletwo}.)

We present a call-by-name reduction relation for this system, and prove that every typable term reduces to a canonical form.  From this, it follows that the system is consistent.  

For the future, we wish to include the equations in $\Omega$, allowing for propositions such as $M =_A N \supset N =_A M$.  We wish to expand the system with universal quantification, and expand it to a 2-dimensional system (with equations between proofs).  We then wish to add more inductive types and more dimensions, getting ever closer to full homotopy type theory.

Another system with many of the same aims is cubical type theory \cite{cchm:cubical}.  The system PHOML is almost a subsystem of cubical type theory.  We can attempt to embed PHOML into cubical type theory,
mapping $\Omega$ to the universe $U$, and an equation $M =_A N$ to either the type $\Path{A}{M}{N}$ or to $\mathrm{Id}\ A\ M\ N$.  However, PHOML has more definitional equalities than the relevant fragment of cubical type theory; that is, there are definitionally equal terms in PHOML that are mapped to terms that are not definitionally equal in cubical type theory.  In particular, $\reff{x}^+ p$ and $p$ are definitionally equal, whereas the terms $\mathrm{comp}^i x [] p$ and $p$ are not definitionally equal in cubical type theory (but they are propositionally equal).  See Section \ref{section:cubical} for more information.

The proofs in this paper have been formalized in Agda.  The formalization is available at \texttt{https://github.com/radams78/TYPES2016}.

\section{Predicative Higher-Order Minimal Logic with Extensional Equality}

We call the following type theory PHOML, or \emph{predicative higher-order minimal logic with extensional equality}.  

\subsection{Syntax}

Fix three disjoint, infinite sets of variables, which we shall call \emph{term variables}, \emph{proof variables}
and \emph{path variables}.  We shall use $x$ and $y$ as term variables, $p$ and $q$ as proof variables,
$e$ as a path variable, and $z$ for a variable that may come from any of these three sets.

The syntax of PHOML is given by the grammar:

$$
\begin{array}{lrcl}
\text{Type} & A,B,C & ::= & \Omega \mid A \rightarrow B \\
\text{Term} & L,M,N, \varphi,\psi,\chi & ::= & x \mid \bot \mid \varphi \supset \psi \mid \lambda x:A.M \mid MN \\
\text{Proof} & \delta, \epsilon & ::= & p \mid \lambda p:\varphi.\delta \mid \delta \epsilon \mid P^+ \mid P^- \\
\text{Path} & P, Q & ::= & e \mid \reff{M} \mid P \supset^* Q \mid \univ{\varphi}{\psi}{P}{Q} \mid \\
& & & \triplelambda e : x =_A y. P \mid P_{MN} Q \\
\text{Context} & \Gamma, \Delta, \Theta & ::= & \langle \rangle \mid \Gamma, x : A \mid \Gamma, p : \varphi \mid \Gamma, e : M =_A N \\
\text{Judgement} & \mathbf{J} & ::= & \Gamma \vald \mid \Gamma \vdash M : A \mid \Gamma \vdash \delta : \varphi \mid \\
& & & \Gamma \vdash P : M =_A N
\end{array}
$$

In the path $\triplelambda e : x =_A y . P$, the term variables $x$ and $y$ must be distinct.  (We also have $x \not\equiv e \not\equiv y$, thanks to our
stipulation that term variables and path variables are disjoint.)  The term variable $x$ is bound within $M$ in the term $\lambda x:A.M$,
and the proof variable $p$ is bound within $\delta$ in $\lambda p:\varphi.\delta$.  The three variables $e$, $x$ and $y$ are bound within $P$ in the path
$\triplelambda e:x =_A y.P$.  We identify terms, proofs and paths up to $\alpha$-conversion.  We write $E[z:=F]$ for the result of substituting $F$ for $z$ within
$E$, using $\alpha$-conversion to avoid variable capture.

We shall use the word 'expression' to mean either a type, term, proof, path, or equation (an equation having the form $M =_A N$).  We shall use $E$, $F$, $S$ and $T$ as metavariables that range over expressions.

Note that we use both Roman letters $M$, $N$ and Greek letters $\varphi$, $\psi$, $\chi$ to range over terms.  Intuitively, a term is understood as either a proposition or a function,
and we shall use Greek letters for terms that are intended to be propositions.  Formally, there is no significance to which letter we choose.

Note also that the types of PHOML are just the simple types over $\Omega$; therefore, no variable can occur in a type.

The intuition behind the new expressions is as follows (see also the rules of deduction in Figure \ref{fig:lambdaoe}).  For any object $M : A$, there is the trivial path $\reff{M} : M =_A M$.  The constructor $\supset^*$ ensures congruence for $\supset$ --- if $P : \varphi =_\Omega \varphi'$ and $Q : \psi =_\Omega \psi'$ then $P \supset^* Q : \varphi \supset \psi =_\Omega \varphi' \supset \psi'$.  The constructor $\mathsf{univ}$ gives 'univalence' (propositional extensionality) for our propositions: if $\delta : \varphi \supset \psi$ and $\epsilon : \psi \supset \varphi$, then $\univ{\varphi}{\psi}{\delta}{\epsilon}$ is a path $\varphi =_\Omega \psi$.  The constructors $^+$ and $^-$ are the converses, which denote the action of transport along a path: if $P$ is a path of type $\varphi =_\Omega \psi$, then $P^+$ is a proof of $\varphi \supset \psi$, and $P^-$ is a proof of $\psi \supset \varphi$.  

The constructor $\triplelambda$ gives functional extensionality.  Let $F$ and $G$ be functions of type $A \rightarrow B$.  If $F x =_B G y$ whenever $x =_A y$, then $F =_{A \rightarrow B} G$.  More formally, if $P$ is a path of type $Fx =_B Gy$ that depends on $x : A$, $y : A$ and $e : x =_A y$, then $\triplelambda e : x =_A y . P$ is a path of type $F =_{A \rightarrow B} G$.  The proofs $P^+$ and $P^-$ represent transport along the path $P$.

Finally, if $P$ is a path $M =_{A \rightarrow B} M'$, and $Q$ is a path $N =_A N'$, then $P_{MN} Q$ is a path $MN =_B M'N'$.

\subsubsection{Substitution and Path Substitution}

Intuitively, if $N$ and $N'$ are equal then $M[x:=N]$ and $M[x:=N']$ should be equal.  To handle this syntactically,
we introduce a notion of \emph{path substitution}.  If $N$, $M$ and $M'$ are terms, $x$ a term variable, and $P$ a path, then we shall define a path $N \{ x := P : M = M' \}$.  The intention is that, if
$\Gamma \vdash P : M =_A M'$ and $\Gamma, x : A \vdash N : B$ then $\Gamma \vdash N \{ x := P : M = M' \} : N [ x:= M ] =_B N [ x := M' ]$ (see Lemma \ref{lm:pathsub}). 

\begin{definition}[Path Substitution]
Given terms $M_1$, \ldots, $M_n$ and $N_1$, \ldots, $N_n$; paths $P_1$, \ldots, $P_n$; term variables $x_1$, \ldots, $x_n$; and a term $L$, define the path $$L \{ x_1 := P_1 : M_1 = N_1 , \ldots, x_n := P_n : M_n = N_n \}$$ as follows.
\begin{align*}
x_i \{ \vec{x} := \vec{P} : \vec{M} = \vec{N} \} & \eqdef P_i \\
y \{ \vec{x} := \vec{P} : \vec{M} = \vec{N} \} & \eqdef \reff{y} \qquad (y \not\equiv x_1, \ldots, x_n) \\
\bot \{ \vec{x} := \vec{P} : \vec{M} = \vec{N} \} & \eqdef \reff{\bot} \\
(LL') \{ \vec{x} := \vec{P} : \vec{M} = \vec{N} \} \\
\omit\rlap{\qquad \qquad $\eqdef L \{ \vec{x} := \vec{P} : \vec{M} = \vec{N} \}_{L' [\vec{x} := \vec{M}] L' [\vec{x} := \vec{N}]} L' \{ \vec{x} := \vec{P} : \vec{M} = \vec{N} \}$} \\
(\lambda y:A.L) \{ \vec{x} := \vec{P} : \vec{M} = \vec{N} \} & \\
\omit\rlap{\qquad\qquad $\eqdef \triplelambda e : a =_A a' . L \{ \vec{x} := \vec{P} : \vec{M} = \vec{N} , y := e : a = a' \}$} \\
(\varphi \supset \psi) \{ \vec{x} := \vec{P} : \vec{M} = \vec{N} \} & \eqdef \varphi \{ \vec{x} := \vec{P} : \vec{M} = \vec{N} \} \supset^* \psi \{ \vec{x} := \vec{P} : \vec{M} = \vec{N} \}
\end{align*}
\end{definition}

We shall often omit the endpoints $\vec{M}$ and $\vec{N}$.

\begin{note}
The case $n = 0$ is permitted, and we shall have that, if $\Gamma \vdash M : A$ then $\Gamma \vdash M \{\} : M =_A M$.  There are thus two paths from a term $M$ to itself: $\reff{M}$ and $M \{\}$.  There are not always equal; for example, $(\lambda x:A.x) \{\} \equiv \triplelambda e : x =_A y. e$, which (after we define the reduction relation) will not be convertible with $\reff{\lambda x:A.x}$.
\end{note}

The following lemma shows how substitution and path substitution interact.

\begin{lemma}
\label{lm:subpathsub}
Let $\vec{y}$ be a sequences of variables and $x$ a distinct variable.  Then
\begin{enumerate}
\label{lm:pathsubsub}
\item
\label{lm:subpathsubi}
$ \begin{aligned}[t]
& M [ x:= N ] \{ \vec{y} := \vec{P} : \vec{L} = \vec{L'} \} \\
& \equiv M \{ x := N \{ \vec{y} := \vec{P} : \vec{L} = \vec{L'} \} : N [ \vec{y}:= \vec{L} ] = N [ \vec{y} := \vec{L'} ], \vec{y} := \vec{P} : \vec{L} = \vec{L'} \}
\end{aligned} $
\item
\label{lm:subpathsubii}
$ \begin{aligned}[t]
& M \{ \vec{y} := \vec{P} : \vec{L} = \vec{L'} \} [ x := N ] \\
& \equiv M \{ \vec{y} := \vec{P} [x := N] : \vec{L} [x := N] = \vec{L'} [x := N], x := \reff{N} : N = N \}
\end{aligned} $
\end{enumerate}
\end{lemma}

\begin{proof}
An easy induction on $M$ in all cases.
\end{proof}

\begin{note}
The familiar substitution lemma also holds as usual: $t [\vec{z_1} := \vec{s_1}] [\vec{z_2} := \vec{s_2}] \equiv t [\vec{z_1} := \vec{s_1}[\vec{z_2} := \vec{s_2}], 
\vec{z_2} := \vec{s_2}]$.  We cannot form a lemma about the fourth case, simplifying $M \{ \vec{x} := \vec{P} \} \{ \vec{y} := \vec{Q} \}$, because
$M \{ \vec{x} := \vec{P} \}$ is a path, and path substitution can only be applied to a term.
\end{note}

We introduce a notation for simultaneous substitution and path substitution of several variables:

\begin{definition}
A \emph{substitution} is a function that maps term variables to terms, proof variables to proofs, and path variables to paths.
We write $E[\sigma]$ for the result of substituting the expression $\sigma(z)$ for $z$ in $E$, for each variable $z$ in the domain of $\sigma$.

A \emph{path substitution} $\tau$ is a function whose domain is a finite set of term variables,
and which maps each term variable to a path.  Given a path substitution $\tau$ and substitutions $\rho$, $\sigma$
with the same domain $\{ x_1, \ldots, x_n \}$, we write
$$ M \{ \tau : \rho = \sigma \} \text{ for } M \{ x_1 := \tau(x_1) : \rho(x_1) = \sigma(x_1), \ldots, \tau(x_n) : \rho(x_n) = \sigma(x_n) \} \enspace . $$
\end{definition}

\subsubsection{Call-By-Name Reduction}

\begin{definition}[Call-By-Name Reduction]
Define the relation of \emph{call-by-name reduction} $\rightarrow$ on the expressions.  The inductive
definition is given by the rules in Figure \ref{fig:reduction}

\begin{figure}
\paragraph*{Reduction on Terms}
$$ \infer{(\lambda x:A.M)N \rightarrow M[x:=N]}{} \quad
\infer{MN \rightarrow M'N}{M \rightarrow M'} $$
$$ \infer{\varphi \supset \psi \rightarrow \varphi' \supset \psi}{\varphi \rightarrow \varphi'} \quad
\infer{\varphi \supset \psi \rightarrow \varphi \supset \psi'}{\psi \rightarrow \psi'} $$
\paragraph*{Reduction on Proofs}
$$\infer{(\lambda p : \varphi . \delta)\epsilon \rightarrow \delta [ p := \epsilon ]}{} \quad
\infer{\reff{\varphi}^+ \rightarrow \lambda p : \varphi . p}{} \quad
\infer{\reff{\varphi}^- \rightarrow \lambda p : \varphi . p}{} $$
$$ \infer{\delta \epsilon \rightarrow \delta' \epsilon}{\delta \rightarrow \delta'} \quad
\infer{\univ{\varphi}{\psi}{\delta}{\epsilon}^+ \rightarrow \delta}{} \quad
\infer{\univ{\varphi}{\psi}{\delta}{\epsilon}^- \rightarrow \epsilon}{} $$
$$ \infer{P^+ \rightarrow Q^+}{P \rightarrow Q} \quad
\infer{P^- \rightarrow Q^-}{P \rightarrow Q}$$
\paragraph*{Reduction on Paths}
$$\infer{(\triplelambda e:x =_A y.P)_{MN} Q \rightarrow P [x := M, y := N, e := Q]}{} $$
$$ \infer{\reff{\lambda x:A.M}_{N N'} P \rightarrow M \{ x:=P : N = N' \}}{} $$
$$ \infer{\reff{\varphi} \supset^* \reff{\psi} \rightarrow \reff{\varphi \supset \psi}}{} $$
$$ \infer{\reff{\varphi} \supset^* \univ{\psi}{\chi}{\delta}{\epsilon} \rightarrow 
\univ{\varphi \supset \psi}{\varphi \supset \chi}{\lambda p:\varphi \supset \psi. \lambda q : \varphi. \delta (pq)}{\lambda p : \varphi \supset \chi. \lambda q : \varphi. \epsilon (pq)}}{} $$
$$ \infer{\univ{\varphi}{\psi}{\delta}{\epsilon} \supset^* \reff{\chi} \rightarrow
\univ{\varphi \supset \chi}{\psi \supset \chi}{\lambda p : \varphi \supset \chi. \lambda q : \psi. p (\epsilon q)}{\lambda p:\psi \supset \chi. \lambda q : \varphi. p(\delta q)}}{} $$
$$ \infer{\begin{array}{l}
\univ{\varphi}{\psi}{\delta}{\epsilon} \supset^* \univ{\varphi'}{\psi'}{\delta'}{\epsilon'} \\
 \rightarrow
\univ{\varphi \supset \varphi'}{\psi \supset \psi'}{\lambda p : \varphi \supset \varphi'. \lambda q : \psi. \delta' (p(\epsilon q))}{\lambda p : \psi \supset \psi'. \lambda q : \varphi. \epsilon' (p (\delta q))}
\end{array}}{} $$
$$ \infer{P_{MN} Q \rightarrow P'_{MN} Q}{P \rightarrow P'} \quad
\infer{\reff{M}_{NN'}P \rightarrow \reff{M'}_{NN'}P}{M \rightarrow M'} $$
$$ \infer{P \supset^* Q \rightarrow P' \supset^* Q}{P \rightarrow P'} \quad
\infer{P \supset^* Q \rightarrow P \supset^* Q'}{Q \rightarrow Q'} $$
\caption{Reduction in PHOML}
\label{fig:reduction}
\end{figure}
\end{definition}

\begin{lemma}[Confluence]
\label{lm:diamond}
If $E \twoheadrightarrow F$ and $E \twoheadrightarrow G$, then there exists $H$ such that $F \twoheadrightarrow H$ and $G \twoheadrightarrow H$.
\end{lemma}

\begin{proof}
The proof is given in Appendix \ref{section:confluence}.
\end{proof}

\begin{lemma}[Reduction respects path substitution]
\label{lm:resp-sub}
If $M \rightarrow N$ then $M \{ \tau : \rho = \sigma \} \rightarrow N \{ \tau : \rho = \sigma \}$.
\end{lemma}

\begin{proof}
Induction on $M \rightarrow N$.  The only difficult case is $\beta$-contraction.  We have
\begin{align*}
& ((\lambda x:A.M)N)\{ \tau : \rho = \sigma \} \\
\equiv & (\triplelambda e : x =_A x' . M \{ \tau : \rho = \sigma , x := e : x = x' \})_{N [ \rho ] N [ \sigma ]} N \{ \tau : \rho = \sigma \} \\
\rightarrow & M \{ \tau : \rho = \sigma, x := N \{ \tau \} : N [ \rho ] = N [ \sigma ] \} \\
\equiv & M [ x := N ] \{ \tau : \rho = \sigma \} & (\text{Lemma \ref{lm:pathsubsub}})
\end{align*}
\end{proof}

We write $\rightarrow^?$ for the reflexive closure of $\rightarrow$, 
we write $\twoheadrightarrow$ for the reflexive transitive closure of $\rightarrow$, and we write $\simeq$ for the reflexive symmetric transitive closure of $\rightarrow$.
We say an expression $E$ is in \emph{normal form} iff there is no expression $F$ such that $E \rightarrow F$.

\begin{note}
$ $
\begin{enumerate}
\item
Reduction on proofs and paths does \emph{not} respect substitution.  For example, let $M \equiv \lambda x:\Omega.x$.  Then we have
\begin{align}
\reff{\lambda y : \Omega. y'}_{\bot \bot} \reff{\bot} & \rightarrow y' \{ y:= \reff{\bot} : \bot = \bot \} \equiv \reff{y'} \nonumber \\
(\reff{\lambda y : \Omega.y'}_{\bot \bot} \reff{\bot}) [y' := M] & \equiv \reff{\lambda x : \Omega.M}_{\bot \bot} \reff{\bot}
\label{eq:exp1} \\
\label{eq:exp2}
\reff{y'}[y':=M] & \equiv \reff{M} \equiv \reff{\lambda x : \Omega.x}
\end{align}
Expression (\ref{eq:exp1}) does not reduce to (\ref{eq:exp2}).  Instead, (\ref{eq:exp1}) reduces to
$$ M \{ y := \reff{\bot} : \bot = \bot \} \equiv \triplelambda e : x =_\Omega x'. e \enspace . $$
\item
Reduction on terms does respect substitution: if $M \rightarrow N$ then $M[x:=P] \rightarrow N[x:=P]$, as is easily shown by induction on $M \rightarrow N$.
\end{enumerate}
\end{note}

\subsection{Rules of Deduction}

The rules of deduction of PHOML are given in Figure \ref{fig:lambdaoe}.

\newcommand{\RvarT}{\ensuremath(\mathsf{varT})}
\begin{figure}
\paragraph*{Contexts}
$$ (\langle \rangle) \quad \vcenter{\infer{\langle \rangle \vald}{}} \qquad
(\mathrm{ctx}_T) \quad \vcenter{\infer{\Gamma, x : A \vald}{\Gamma \vald}} \qquad 
(\mathrm{ctx}_P) \quad \vcenter{\infer{\Gamma, p : \varphi \vald}{\Gamma \vdash \varphi : \Omega}} $$
$$ (\mathrm{ctx}_E) \quad \vcenter{\infer{\Gamma, e : M =_A N \vald}{\Gamma \vdash M : A \quad \Gamma \vdash N : A}} $$
$$ (\mathrm{var}_T) \quad \vcenter{\infer[(x : A \in \Gamma)]{\Gamma \vdash x : A}{\Gamma \vald}} \qquad
(\mathrm{var}_P) \quad \vcenter{\infer[(p : \varphi \in \Gamma)]{\Gamma \vdash p : \varphi}{\Gamma \vald}} $$
$$ (\mathrm{var}_E) \quad \vcenter{\infer[(e : M =_A N \in \Gamma)]{\Gamma \vdash e : M =_A N}{\Gamma \vald}} $$

\paragraph*{Terms}
$$ (\bot) \quad \vcenter{\infer{\Gamma \vdash \bot : \Omega}{\Gamma \vald}} \qquad
(\supset) \quad \vcenter{\infer{\Gamma \vdash \varphi \supset \psi : \Omega}{\Gamma \vdash \varphi : \Omega \quad \Gamma \vdash \psi : \Omega}} $$
$$ (\mathrm{app}_T) \quad \vcenter{\infer{\Gamma \vdash M N : B} {\Gamma \vdash M : A \rightarrow B \quad \Gamma \vdash N : A}} \qquad
(\lambda_T) \quad \vcenter{\infer{\Gamma \vdash \lambda x:A.M : A \rightarrow B}{\Gamma, x : A \vdash M : B}} $$

\paragraph*{Proofs}
$$ (\mathrm{app}_P) \quad \vcenter{\infer{\Gamma \vdash \delta \epsilon : \psi} {\Gamma \vdash \delta : \varphi \supset \psi \quad \Gamma \vdash \epsilon : \varphi}} \qquad
(\lambda_P) \quad \vcenter{\infer{\Gamma \vdash \lambda p : \varphi . \delta : \varphi \supset \psi}{\Gamma, p : \varphi \vdash \delta : \psi}} $$
$$ (\mathrm{conv}_P) \quad \vcenter{\infer[(\varphi \simeq \psi)]{\Gamma \vdash \delta : \psi}{\Gamma \vdash \delta : \varphi \quad \Gamma \vdash \psi : \Omega}} $$

\paragraph*{Paths}
$$ (\mathrm{ref}) \quad \vcenter{\infer{\Gamma \vdash \reff{M} : M =_A M}{\Gamma \vdash M : A}}
\qquad
(\supset^*) \quad \vcenter{\infer{\Gamma \vdash P \supset^* Q : \varphi \supset \psi =_\Omega \varphi' \supset \psi'}{\Gamma \vdash P : \varphi =_\Omega \varphi' \quad \Gamma \vdash Q : \psi =_\Omega \psi'}} $$
$$ (\mathrm{univ}) \quad \vcenter{\infer{\Gamma \vdash \univ{\varphi}{\psi}{\delta}{\epsilon} : \varphi =_\Omega \psi}{\Gamma \vdash \delta : \varphi \supset \psi \quad \Gamma \vdash \epsilon : \psi \supset \varphi}} $$
$$ (\mathrm{plus}) \quad \vcenter{\infer{\Gamma \vdash P^+ : \varphi \supset \psi}{\Gamma \vdash P : \varphi =_\Omega \psi}}
\qquad
(\mathrm{minus}) \quad \vcenter{\infer{\Gamma \vdash P^- : \psi \supset \varphi}{\Gamma \vdash P : \psi =_\Omega \psi}} $$
$$ (\triplelambda) \quad \vcenter{\infer{\Gamma \vdash \triplelambda e : x =_A y . P : M =_{A \rightarrow B} N}
  {\begin{array}{c}
     \Gamma, x : A, y : A, e : x =_A y \vdash P : M x =_B N y \\
     \Gamma \vdash M : A \rightarrow B \quad
\Gamma \vdash N : A \rightarrow B
     \end{array}}} $$
$$ (\mathrm{app}_E) \quad \vcenter{\infer{\Gamma \vdash P_{NN'}Q : MN =_B M' N'}{\Gamma \vdash P : M =_{A \rightarrow B} M' \quad \Gamma \vdash Q : N =_A N' \quad \Gamma \vdash N : A \quad \Gamma \vdash N' : A}} $$
$$ (\mathrm{conv}_E) \quad \vcenter{\infer[(M \simeq M', N \simeq N')]{\Gamma \vdash P : M' =_A N'}{\Gamma \vdash P : M =_A N \quad \Gamma \vdash M' : A \quad \Gamma \vdash N' : A}} $$
\caption{Rules of Deduction of $\lambda oe$}
\label{fig:lambdaoe}
\end{figure}

\subsubsection{Metatheorems}

\label{section:meta}

In the lemmas that follow, the letter $\mathcal{J}$ stands for any of the expressions that may occur to the right of the turnstile in a judgement, i.e.~$\mathrm{valid}$, $M : A$, $\delta : \varphi$, or $P : M =_A N$.

\begin{lemma}[Context Validity]
Every derivation of $\Gamma, \Delta \vdash \mathcal{J}$ has a subderivation of $\Gamma \vald$.
\end{lemma}

\begin{proof}
Induction on derivations.
\end{proof}

\begin{lemma}[Weakening]
If $\Gamma \vdash \mathcal{J}$, $\Gamma \subseteq \Delta$ and $\Delta \vald$ then $\Delta \vdash \mathcal{J}$.
\end{lemma}

\begin{proof}
Induction on derivations.
\end{proof}

\begin{lemma}[Type Validity]
$ $
\begin{enumerate}
\item
If $\Gamma \vdash \delta : \varphi$ then $\Gamma \vdash \varphi : \Omega$.
\item
If $\Gamma \vdash P : M =_A N$ then $\Gamma \vdash M : A$ and $\Gamma \vdash N : A$.
\end{enumerate}
\end{lemma}

\begin{proof}
Induction on derivations.  The cases where $\delta$ or $P$ is a variable use Context Validity.
\end{proof}

\begin{lemma}[Generation]
$ $
\begin{enumerate}
\item
If $\Gamma \vdash x : A$ then $x : A \in \Gamma$.
\item
If $\Gamma \vdash \bot : A$ then $A \equiv \Omega$.
\item
If $\Gamma \vdash \varphi \supset \psi : A$ then $\Gamma \vdash \varphi : \Omega$, $\Gamma \vdash \psi : \Omega$ and $A \equiv \Omega$.
\item
If $\Gamma \vdash \lambda x:A.M : B$ then there exists $C$ such that $\Gamma, x : A \vdash M : C$ and $B \equiv A \rightarrow C$.
\item
If $\Gamma \vdash MN : A$ then there exists $B$ such that $\Gamma \vdash M : B \rightarrow A$ and $\Gamma \vdash N : B$.
\item
If $\Gamma \vdash p : \varphi$, then there exists $\psi$ such that $p : \psi \in \Gamma$ and $\varphi \simeq \psi$.
\item
If $\Gamma \vdash \lambda p:\varphi.\delta : \psi$, then there exists $\chi$ such that $\Gamma, p : \varphi \vdash \delta : \chi$ and $\psi \simeq (\varphi \supset \chi)$.
\item
If $\Gamma \vdash \delta \epsilon : \varphi$ then there exists $\psi$ such that $\Gamma \vdash \delta : \psi \supset \varphi$ and $\Gamma \vdash \epsilon : \psi$.
\item
If $\Gamma \vdash e : M =_A N$, then there exist $M'$, $N'$ such that $e : M' =_A N' \in \Gamma$ and $M \simeq M'$, $N \simeq N'$.
\item
If $\Gamma \vdash \reff{M} : N =_A P$, then we have $\Gamma \vdash M : A$ and $M \simeq N \simeq P$.
\item
If $\Gamma \vdash P \supset^* Q : \varphi =_A \psi$, then there exist $\varphi_1$, $\varphi_2$, $\psi_1$, $\psi_2$ such that
$\Gamma \vdash P : \varphi_1 =_\Omega \psi_1$, $\Gamma \vdash Q : \varphi_2 =_\Omega \psi_2$, $\varphi \simeq (\varphi_1 \supset \psi_1)$, $\psi \simeq (\varphi_2 \supset \psi_2)$, and $A \equiv \Omega$.
\item
If $\Gamma \vdash \univ{\varphi}{\psi}{\delta}{\epsilon} : \chi =_A \theta$, then we have $\Gamma \vdash \delta : \varphi \supset \psi$, $\Gamma \vdash \epsilon : \psi \supset \varphi$,
$\chi \simeq \varphi$, $\theta \simeq \psi$ and $A \equiv \Omega$.
\item
If $\Gamma \vdash \triplelambda e : x =_A y. P : M =_B N$ then there exists $C$ such that $\Gamma, x : A, y : A, e : x =_A y \vdash P : M x =_C N y$
and $B \equiv A \rightarrow C$.
\item
If $\Gamma \vdash P_{M M'} Q : N =_A N'$, then there exist $B$, $F$ and $G$ such that $\Gamma \vdash P : F =_{B \rightarrow A} G$, $\Gamma \vdash Q : M =_B M'$, $N \simeq F M$
and $N' \simeq G M'$.
\item
If $\Gamma \vdash P^+ : \varphi$, then there exist $\psi$, $\chi$ such that $\Gamma \vdash P : \psi =_\Omega \chi$ and $\varphi \simeq (\psi \supset \chi)$.
\item
If $\Gamma \vdash P^- : \varphi$, there exist $\psi$, $\chi$ such that $\Gamma \vdash P : \psi =_\Omega \chi$ and $\varphi \simeq (\chi \supset \psi)$.
\end{enumerate}
\end{lemma}

\begin{proof}
Induction on derivations.
\end{proof}

\subsubsection{Substitutions}

\begin{definition}
Let $\Gamma$ and $\Delta$ be contexts.  A \emph{substitution from $\Delta$ to $\Gamma$}\footnote{These have also been called \emph{context morphisms}, for example in Hoffman \cite{Hofmann97syntaxand}.},
$\sigma : \Delta \Rightarrow \Gamma$,
is a substitution whose domain is $\dom \Gamma$ such that:
\begin{itemize}
\item
for every term variable $x : A \in \Gamma$, we have $\Delta \vdash \sigma(x) : A$;
\item
for every proof variable $p : \varphi \in \Gamma$, we have $\Delta \vdash \sigma(p) : \varphi [ \sigma ]$;
\item
for every path variable $e : M =_A N \in \Gamma$, we have $\Delta \vdash \sigma(e) : M [ \sigma ] =_A N [ \sigma ]$.
\end{itemize}
\end{definition}

\begin{lemma}[Well-Typed Substitution]
If $\Gamma \vdash \mathcal{J}$, $\sigma : \Delta \Rightarrow \Gamma$ and $\Delta \vald$, then $\Delta \vdash \mathcal{J} [\sigma]$.
\end{lemma}

\begin{proof}
Induction on derivations.
\end{proof}

\begin{definition}
If $\rho, \sigma : \Delta \Rightarrow \Gamma$ and $\tau$ is a path substitution whose domain
is the term variables in $\dom \Gamma$, then we write
$\tau : \sigma = \rho : \Delta \Rightarrow \Gamma$ iff, for each variable $x : A \in \Gamma$, we have
$\Delta \vdash \tau(x) : \sigma(x) =_A \rho(x)$.
\end{definition}

\begin{lemma}[Path Substitution]
\label{lm:pathsub}
If $\tau : \sigma = \rho : \Delta \Rightarrow \Gamma$ and $\Gamma \vdash M : A$ and $\Delta \vald$,
then $\Delta \vdash M \{ \tau : \sigma = \rho \} : M [ \sigma ] =_A M [ \rho ]$.
\end{lemma}

\begin{proof}
Induction on derivations.
\end{proof}

\begin{proposition}[Subject Reduction]
If $\Gamma \vdash s : T$ and $s \twoheadrightarrow t$ then $\Gamma \vdash t : T$.
\end{proposition}

\begin{proof}
It is sufficient to prove the case $s \rightarrow t$.  The proof is by a case analysis on $s \rightarrow t$, using the Generation,
Well-Typed Substitution and Path Substitution Lemmas.
\end{proof}

\subsubsection{Canonicity}

\begin{definition}[Canonical Object]
$ $
\begin{itemize}
\item
The \emph{canonical propositions}, are given by the grammar
$$ \theta ::= \bot \mid \theta \supset \theta $$
\item
A \emph{canonical proof} is one of the form $\lambda p : \varphi . \delta$.
\item
A \emph{canonical path} is one of the form $\reff{M}$, $\univ{\phi}{\psi}{\delta}{\epsilon}$ or $\triplelambda e : x =_A y.P$.
\end{itemize}
\end{definition}

\begin{lemma}
\label{lm:compat-beta}
Suppose $\varphi$ reduces to a canonical proposition $\theta$, and $\varphi \simeq \psi$.  Then $\psi$ reduces to $\theta$.
\end{lemma}

\begin{proof}
This follows from the fact that $\rightarrow$ satisfies the diamond property, and every canonical proposition $\theta$ is a normal form.
\end{proof}

\subsubsection{Neutral Expressions}

\begin{definition}[Neutral]$ $
The \emph{neutral} terms, paths and proofs are given by the grammar
$$ \begin{array}{lrcl}
\text{Neutral term} & M_n & ::= & x \mid M_n N \\
\text{Neutral proof} & \delta_n & ::= & p \mid P_n^+ \mid P_n^- \mid \delta_n \epsilon \\
\text{Neutral path} & P_n & ::= & e \mid P_n \supset^* Q \mid Q \supset^* P_n \mid (P_n)_{MN} Q
\end{array} $$
\end{definition}

\section{Examples}

We present two examples illustrating the way that proofs and paths behave in PHOML.  In each case, we compare the example with the same construction performed in cubical type theory.

\subsection{Functions Respect Logical Equivalence}
\label{section:exampletwo}

As discussed in the introduction, every function of type $\Omega \rightarrow \Omega$ that can be constructed in PHOML must respect logical equivalence.  This fact can actually be proved in PHOML,
in the following sense: there exists a proof $\delta$ of
$$ f : \Omega \rightarrow \Omega, x : \Omega, y : \Omega, p : x \supset y, q : y \supset x \vdash \delta : f x \supset f y $$
and a proof of $f y \supset f x$ in the same context.  Together, these can be read as a proof of 'if $f : \Omega \rightarrow \Omega$ and $x$ and $y$ are logically equivalent, then $fx$ and $fy$ are logically equivalent'.

Specifically, take
$$ \delta \eqdef (\reff{f}_{xy} \univ{x}{y}{p}{q})^+ \enspace . $$

Note that this is not possible in Martin-L\"{o}f Type Theory.

In cubical type theory, we can construct a term $\delta$ such that
$$ f : \Prop \rightarrow \Prop, x : \Prop, y : \Prop, p : x.1 \rightarrow y.1, q : y.1 \rightarrow x.1 \vdash \delta : (f x).1 \rightarrow (f y).1 $$
In fact, we can go further and prove that equality of propositions is equal to logical equivalence.  That is, we can prove
$$ \Path{U}{(\Path{\Prop}{x}{y})}{((x.1 \rightarrow y.1) \times (y.1 \rightarrow x.1))} \enspace . $$

\subsection{Computation with Paths}

Let $\top \eqdef \bot \supset \bot$.
Using propositional extensionality, we can construct a path of type $\top = \top \supset \top$, and hence a proof of $\top \supset (\top \supset \top)$.
But which of the canonical proofs of $\top \supset (\top \supset \top)$ have we constructed?

We define
$$ \top := \bot \supset \bot, \quad \iota := \lambda p:\bot.p, \quad I := \lambda x:\Omega.x, \quad F := \lambda x:\Omega.\top \supset x, \quad H := \lambda h.h \top \enspace . $$

Let $\Gamma$ be the context
$$ \Gamma \eqdef x : \Omega, y : \Omega, e : x =_\Omega y \enspace . $$

Then we have
\begin{align*}
\Gamma & \vdash \lambda p:\top \supset x. e^+ (p \iota) & : & (\top \supset x) \supset y \\
\Gamma & \vdash \lambda m : y. \lambda n : \top. e^- m & : & y \supset (\top \supset x) \\
\Gamma & \vdash \univ*{\lambda p:\top \supset x. e^- m}{\lambda m:y. \lambda n:\top. e^- m} & : & (\top \supset x) =_\Omega y
\end{align*}

Let $P \equiv \univ*{\lambda p:\top \supset x. e^+ (p \iota)}{\lambda m:y. \lambda n:\top. e^- m}$.  Then

\begin{align}
\therefore & \vdash \triplelambda e:x =_\Omega y. P & : & F =_{\Omega \rightarrow \Omega} I \label{eq:llleP} \\
\therefore & \vdash (\reff{H})_{FI}(\triplelambda e:x =_\Omega y. P) & : & (\top \supset \top) =_\Omega \top \label{eq:llleP2} \\
\therefore & \vdash ((\reff{H})_{FI}(\triplelambda e:x =_\Omega y.P))^- & : & \top \supset (\top \supset \top) \label{eq:llleP3}
\end{align}

And now we compute:

\begin{align*}
& ((\reff{H})_{FI}(\triplelambda e:x =_\Omega y.P))^- \\
\twoheadrightarrow & ((h \top) \{ h := \triplelambda e : x =_\Omega y.P : F = I \})^- \\
\equiv & ((\triplelambda e : x =_\Omega y.P)_{\top \top} (\reff{\top}))^- \\
\rightarrow & (P [ x := \top, y := \top, e := \reff{\top} ])^- \\
\equiv & \univ*{\lambda p : \top \supset \top. \reff{\top}^+ (p \iota)}{\lambda m:\top. \lambda n:\top. \reff{\top}^- m}^- \\
\rightarrow & \lambda m : \top. \lambda n : \top. \reff{\top}^- m
\end{align*}

Therefore, given proofs $\delta, \epsilon : \top$, we have
$$ ((\reff{H})_{FI}(\triplelambda e:x =_\Omega y.P))^- \delta \epsilon \twoheadrightarrow \delta \enspace . $$

Thus, the construction gives a proof of $\top \supset (\top \supset \top)$ which, given two proofs of $\top$, selects the first.  We could have anticipated this:
consider the context $\Delta \eqdef X : \Omega, Y : \Omega, p : X$.  By
replacing in our example some occurrences of $\top$ with $X$ and others with $Y$, and replacing $\iota$ with $p$, we can obtain a path
$$ Y =_\Omega X \supset Y $$
and hence a proof of $Y \supset (X \supset Y)$.  By parametricity, any proof that we can construct in the context $\Delta$ of this proposition must return the left input.

\subsubsection{Comparison with Cubical Type Theory}
\label{section:cubical}

In cubical type theory, we say that a type $A$ is a \emph{proposition} iff any two terms of type $A$ are propositionally equal; that is, there exists a path between any two terms of type $A$.  Let
$$ \isProp{A} \eqdef \Pi x,y:A. \Path{A}{x}{y} $$
and let $\Prop$ be the type of all types in $U$ that are propositions:
$$ \Prop \eqdef \Sigma X : U. \isProp{X} \enspace . $$
Let $\bot$ be any type in the universe $U$ that is a proposition; that is, there exists a term of type $\isProp{\bot}$.  ($\bot$ may be the empty type, but we do not require this in what follows.)

Define
$$ \top := \bot \rightarrow \bot $$
Then there exists a term $\top_\Prop$ of type $\isProp{\top}$ (we omit the details).  Define
$$ I := \lambda X:\Prop.X.1, \quad F := \lambda X : \Prop.\top \rightarrow X.1, \quad H := \lambda h.h (\top , \top_\Prop) $$
Then we have
$$ \vdash \top : U \quad
\vdash I  : \Prop \rightarrow U \quad
\vdash F  : \Prop \rightarrow U \quad
\vdash H  : (\Prop \rightarrow U) \rightarrow U $$

From the fact that univalence is provable in cubical type theory \cite{cchm:cubical}, we can construct a term $Q$ such that
$$ \vdash Q : \Path{(\Prop \rightarrow U)}{I}{F} \enspace . $$
Hence we have
$$ \vdash \langle i \rangle H (Q i) : \Path{U}{HI}{HF} $$
which is definitionally equal to
$$ \vdash \langle i \rangle H (Q i) : \Path{U}{\top \rightarrow \top}{\top} $$
From this, we can apply transport to create a term $Q' : \top \rightarrow \top \rightarrow \top$.  Applying this to any terms $\delta, \epsilon : \top$ gives a term that
is definitionally equal to
$$ Q' \delta \epsilon = \mapid{\top}{\mapid{\top}{\delta}} $$
where $\mapid{}{}$ represents transport across the trivial path:
$$ \mapid{A}{t} \eqdef \comp^i \, A \, [] \, t \qquad (i \text{ does not occur in } A) \enspace . $$
(For the details of the calculation, see Appendix \ref{appendix:cubical}.)

In the version of cubical type theory given in \cite{bch:cubical}, we have $\mapid{X}{x}$ is definitionally equal to $x$, and therefore $Q' \delta \epsilon = \delta$, just as in PHOML.
This is no longer true in the version of cubical type theory given in \cite{cchm:cubical}.

\section{Computable Expressions}
\label{section:computable}

We now proceed with the proof of canonicity for PHOML.  Our proof follows the lines of the Girard-Tait reducibility method \cite{tait1967}: we define what it means to be a \emph{computable}
term (proof, path) of a given type (proposition, equation), and prove: (1) every typable expression is computable (2) every computable expression reduces to either a neutral or a canonical
expression.  In particular, every closed computable expression reduces to a canonical expression.

In this section, we use $E$, $F$, $S$ and $T$ as metavariables that range over expressions.  In each case, either $E$ and $F$ are terms and $S$ and $T$ are types; or $E$ and $F$ are
proofs and $S$ and $T$ are propositions; or $E$ and $F$ are paths and $S$ and $T$ are equations.

\begin{definition}[Computable Expression]
We define the relation $\models E : T$, read `$E$ is a computable expression of type $T$', as follows.
\begin{itemize}
\item
$\models \delta : \bot$ iff $\delta$ reduces to a neutral proof.
\item
For $\theta$ and $\theta'$ canonical propositions, $\models \delta : \theta \supset \theta'$ iff, for all $\epsilon$ such that $\models \epsilon : \theta$, we have $\models \delta \epsilon : \theta'$.
\item
If $\varphi$ reduces to the canonical proposition $\theta$, then $\models \delta : \varphi$ iff $\models \delta : \theta$.
\item
$\models P : \varphi =_\Omega \psi$ iff $\models P^+ : \varphi \supset \psi$ and $\models P^- : \psi \supset \varphi$.
\item
$\models P : M =_{A \rightarrow B} M'$ iff, for all $Q$, $N$, $N'$ such that $\models N : A$ and $\models N' : A$ and $\models Q : N =_A N'$, then we have $\models P_{NN'}Q : MN =_B M'N'$.
\item
$\models M : A$ iff $\models M \{\} : M =_A M$.
\end{itemize}
Note that the last three clauses define $\models M : A$ and $\models P : M =_A N$ simultaneously by recursion on $A$.
\end{definition}

\begin{definition}[Computable Substitution]
Let $\sigma$ be a substitution with domain $\dom \Gamma$.  We write $\models \sigma : \Gamma$ and say that
$\sigma$ is a \emph{computable} substitution on $\Gamma$ iff, for every entry $z : T$ in $\Gamma$, we have $\models \sigma(z) : T [ \sigma ]$.

We write $\models \tau : \rho = \sigma : \Gamma $, and say $\tau$ is a \emph{computable} path substitution between $\rho$ and $\sigma$, iff, for every term variable entry $x : A$ in $\Gamma$, we have $\models \tau(x) : \rho(x) =_A \sigma(x)$.
\end{definition}

\begin{lemma}[Conversion]
\label{lm:conv-compute}
If $\models E : S$ and $S \simeq T$ then $\models E : T$.
\end{lemma}

\begin{proof}
This follows easily from the definition and Lemma \ref{lm:compat-beta}.
\end{proof}

\begin{lemma}[Expansion]
\label{lm:expansion}
If $\models F : T$ and $E \rightarrow F$ then $\models E : T$.
\end{lemma}

\begin{proof}
An easy induction, using the fact that call-by-name reduction respects path substitution (Lemma \ref{lm:resp-sub}).
\end{proof}

\begin{lemma}[Reduction]
\label{lm:reduction}
If $\models E : T$ and $E \rightarrow F$ then $\models F : T$.
\end{lemma}

\begin{proof}
An easy induction, using the fact that call-by-name reduction is confluent (Lemma \ref{lm:diamond}).
\end{proof}

\begin{definition}
We introduce a closed term $c_A$ for every type $A$ such that $\models c_A : A$.
\begin{align*}
c_\Omega & \eqdef \bot \\
c_{A \rightarrow B} & \eqdef \lambda x:A.c_B
\end{align*}
\end{definition}

\begin{lemma}
$\models c_A : A$
\end{lemma}

\begin{proof}
An easy induction on $A$.
\end{proof}

\begin{lemma}[Weak Normalization]
\label{lm:neutral-canon}
$ $
\begin{enumerate}
\item
If $\models \delta : \phi$ then $\delta$ reduces to either a neutral proof or canonical proof.
\item
If $\models P : M =_A N$ then $P$ reduces either to a neutral path or canonical path.
\item
If $\models M : A$ then $M$ reduces either to a canonical proposition or a $\lambda$-term.
\end{enumerate}
\end{lemma}

\begin{proof}
We prove by induction on the canonical proposition $\theta$ that, if $\models \delta : \theta$, then $\delta$ reduces to a neutral proof or
a canonical proof of $\theta$.

If $\models \delta : \bot$ then $\delta$ reduces to a neutral proof.  Now, suppose $\models \delta : \theta \supset \theta'$.  Then $\models \delta p :
\theta'$, so $\delta p$ reduces to either a neutral proof or canonical proof by the induction hypothesis.  This reduction must proceed either by
reducing $\delta$ to a neutral proof, or reducing $\delta$ to a $\lambda$-proof then $\beta$-reducing.

We then prove by induction on the type $A$ that, if $\models P : M =_A N$, then $P$ reduces to a neutral path or a canonical path. The two cases are straightforward.

Now, suppose $\models M : A$, i.e. $\models M \{\} : M =_A M$.  Let $A \equiv A_1 \rightarrow \cdots \rightarrow A_n \rightarrow \Omega$.  Then
$$ \models M \{\}_{c_{A_1} c_{A_1}} c_{A_1} \{\}_{c_{A_2} c_{A_2}} \cdots c_{A_n} \{\}
: M c_{A_1} \cdots c_{A_n} =_\Omega M c_{A_1} \cdots c_{A_n} \enspace . $$
Therefore, $M c_{A_1} \cdots c_{A_n}$ reduces to a canonical proposition.  The reduction must consist either in reducing $M$ to a canonical proposition (if $n = 0$),
or reducing $M$ to a $\lambda$-expression then performing a $\beta$-reduction.
\end{proof}

\begin{lemma}
\label{lm:pre-ref-compute}
If $\models M : A \rightarrow B$ then $M$ reduces to a $\lambda$-expression.
\end{lemma}

\begin{proof}
Similar to the last paragraph of the previous proof.
\end{proof}

\begin{lemma}
\label{lm:ref-compute-Omega}
For any term $\varphi$ that reduces to a canonical proposition, we have $\models \reff{\varphi} : \varphi =_\Omega \varphi$.
\end{lemma}

\begin{proof}
In fact we prove that, for any terms $M$ and $\varphi$ such that $\varphi$ reduces to a canonical proposition, we have
$\models \reff{M} : \varphi =_\Omega \varphi$.

It is sufficient to prove the case where $\varphi$ is a canonical proposition.
We must show that $\models \reff{M}^+ : \varphi \supset \varphi$ and $\models \reff{M}^- : \varphi \supset \varphi$.  So let $\models \delta : \varphi$.
Then $\models \reff{M}^+ \delta : \varphi$ and $\models \reff{M}^- \delta : \varphi$ by Expansion (Lemma \ref{lm:expansion}), as required.
\end{proof}

\begin{lemma}
\label{lm:models-canon}
$\models \varphi : \Omega$ if and only if $\varphi$ reduces to a canonical proposition.
\end{lemma}

\begin{proof}
If $\models \varphi : \Omega$ then $\models \varphi \{\}^+ : \varphi \supset \varphi$.  Therefore $\varphi \supset \varphi$ reduces to a canonical proposition,
and so $\varphi$ must reduce to a canonical proposition.

Conversely, suppose $\varphi$ reduces to a canonical proposition $\theta$.  We have $\varphi \{\} \twoheadrightarrow \theta \{\}$, and $\theta \{\} \twoheadrightarrow \reff{\theta}$ for every canonical proposition $\theta$.  Therefore, $\models \varphi \{\} : \varphi =_\Omega \varphi$ by Expansion (Lemma \ref{lm:expansion}).  Hence $\models \varphi : \Omega$.
\end{proof}

\begin{lemma}
\label{lm:neutral-proof}
If $\delta$ is a neutral proof and $\varphi$ reduces to a canonical proposition, then $\models \delta : \varphi$.
\end{lemma}

\begin{proof}
It is sufficient to prove the case where $\varphi$ is a canonical proposition.  The proof is by induction on $\varphi$.

If $\varphi \equiv \bot$, then $\models \delta : \bot$ immediately from the definition.

If $\varphi \equiv \psi \supset \chi$, then let $\models \epsilon : \psi$.  We have that $\delta \epsilon$ is neutral,
hence $\models \delta \epsilon : \chi$ by the induction hypothesis.
\end{proof}

\begin{lemma}
\label{lm:neutral-path}
Let $\models M : A$ and $\models N : A$.  If $P$ is a neutral path, then $\models P : M =_A N$.
\end{lemma}

\begin{proof}
The proof is by induction on $A$.

For $A \equiv \Omega$: we have that $P^+$ and $P^-$ are neutral proofs, and $M$ and $N$ reduce to canonical propositions (by Lemma \ref{lm:models-canon}), so $\models P^+ : M \supset N$ and
$\models P^- : N \supset M$ by Lemma \ref{lm:neutral-proof}, as required.

For $A \equiv B \rightarrow C$: let $\models L : B$, $\models L' : B$ and $\models Q : L =_B L'$.  Then we have $\models ML : C$, $\models NL' : C$ and
$P_{LL'} Q$ is a neutral path, hence $\models P_{L L'} Q : ML =_C NL'$ by the induction hypothesis, as required.
\end{proof}

\begin{lemma}
\label{lm:ref-compute}
If $\models M : A$ then $\models \reff{M} : M =_A M$.
\end{lemma}

\begin{proof}
If $A \equiv \Omega$, this is just Lemma \ref{lm:ref-compute-Omega}.

So suppose $A \equiv B \rightarrow C$.  Using Lemma \ref{lm:pre-ref-compute}, Reduction (Lemma \ref{lm:reduction}) and Expansion (Lemma \ref{lm:expansion}),
we may assume that $M$ is a $\lambda$-term.  Let $M \equiv \lambda y:D.N$.

Let $\models L : B$ and $\models L' : B$ and $\models P : L =_B L'$.  We must show that
$$ \models \reff{\lambda y:D.N}_{L L'} P : (\lambda y:D.N)L =_C (\lambda y:D.N)L' \enspace . $$
By Expansion and Conversion, it is sufficient to prove
$$ \models N \{ y := P : L = L' \} : N [ y:= L ] =_C N [y := L'] \enspace . $$
We have that $\models (\lambda y:D.N)\{\} : \lambda y:D.N =_{B \rightarrow C} \lambda y:D.N$, and so
$$ \models (\triplelambda e : y =_D y' . N \{ y := e : y = y' \})_{L L'} P : (\lambda y:D.N)L =_C (\lambda y:D.N)L' \enspace , $$
and the result follows by Reduction and Conversion.
\end{proof}

\begin{lemma}
\label{lm:compute-supset*}
If $\models P : \varphi =_\Omega \varphi'$ and $\models Q : \psi =_\Omega \psi'$ then $\models P \supset^* Q : \varphi \supset \psi =_\Omega \varphi' \supset \psi'$.
\end{lemma}

\begin{proof}
By Reduction (Lemma \ref{lm:reduction}) and Expansion (Lemma \ref{lm:expansion}), we may assume that $P$ and $Q$ are either neutral, or have the form $\reff{-}$ or $\univ{-}{-}{-}{-}$ or $\triplelambda e : x =_A y.-$.

We cannot have that $P$ reduces to a $\triplelambda$-path; for let $\varphi'$ reduce to the canonical proposition $\theta_1 \supset \cdots \supset \theta_n \supset \bot$.  Then we have
\[ \models P^+ p q_1 \cdots q_n : \bot \]
and so $P^+ p q_1 \cdots q_n$ must reduce to a neutral path.  Similarly, $Q$ cannot reduce to a $\triplelambda$-path.

If either $P$ or $Q$ is neutral then $P \supset^* Q$ is neutral, and the result follows from Lemma \ref{lm:neutral-path}.

Otherwise, let $\models \delta : \varphi \supset \psi$ and $\epsilon \models \varphi'$.  We must show that $\models (P \supset^* Q)^+ \delta \epsilon : \psi'$.

If $P \equiv \reff{M}$ and $Q \equiv \reff{N}$, then we have
$$ (P \supset^* Q)^+ \delta \epsilon \rightarrow \reff{M \supset N}^+ \delta \epsilon \rightarrow \delta \epsilon \enspace . $$
Now, $\models P^- \epsilon : \varphi$, hence $\models \epsilon : \varphi$ by Reduction, and so $\models \delta \epsilon : \psi$.  Therefore, $\models Q^+ (\delta \epsilon) : \psi'$,
and hence by Reduction $\models \delta \epsilon : \psi'$ as required.

If $P \equiv \reff{M}$ and $Q \equiv \univ{N}{N'}{\chi}{\chi'}$, then we have
\begin{align*}
(P \supset^* Q)^+ \delta \epsilon & \rightarrow \univ{M \supset N}{M \supset N'}{\lambda pq.\chi(pq)}{\lambda pq.\chi'(pq)}^+ \delta \epsilon \\
& \rightarrow (\lambda pq.\chi(pq)) \delta \epsilon \\
& \twoheadrightarrow \chi (\delta \epsilon)
\end{align*}
We have $\models P^- \epsilon : \varphi$, hence $\models \epsilon : \varphi$ by Reduction, and so $\models \delta \epsilon : \psi$.  Therefore,
$\models Q^+ (\delta \epsilon) : \psi'$, and hence by Reduction $\models \chi (\delta \epsilon) : \psi'$ as required.

The other two cases are similar.
\end{proof}

\begin{lemma}
\label{lm:univ-compute}
If $\models \delta : \phi \supset \psi$ and $\models \epsilon : \psi \supset \phi$ then $\models \univ{\phi}{\psi}{\delta}{\epsilon} : \phi =_\Omega \psi$.
\end{lemma}

\begin{proof}
We must show that $\models \univ{\phi}{\psi}{\delta}{\epsilon}^+ : \phi \supset \psi$ and $\models \univ{\phi}{\psi}{\delta}{\epsilon}^- : \psi \supset \phi$.
These follow from the hypotheses, using Expansion (Lemma \ref{lm:expansion}).
\end{proof}

\section{Proof of Canonicity}

\begin{theorem}
$ $
\begin{enumerate}
\item
If $\Gamma \vdash \mathcal{J}$ and $\models \sigma : \Gamma$, then $\models \mathcal{J} [ \sigma ]$.
\item
If $\Gamma \vdash M : A$ and $\models \tau : \rho = \sigma : \Gamma$, then $\models M \{ \tau : \rho = \sigma \} : M [ \rho ] =_A M [ \sigma ]$.
\end{enumerate}
\end{theorem}

\begin{proof}
The proof is by induction on derivations.  Most cases are straightforward, using the lemmas from
Section \ref{section:computable}.  We deal with one case here, the rule ($\lambda_T$).

$$ \infer{\Gamma \vdash \lambda x:A.M : A \rightarrow B}{\Gamma, x : A \vdash M : B} $$

\begin{enumerate}
\item
We must show that
$$ \models \lambda x:A.M[\sigma] : A \rightarrow B \enspace . $$
So let $\models Q : N =_A N'$.  Define the path substitution $\tau$ by
$$ \tau(x) \equiv Q, \qquad \tau(y) \equiv \reff{\sigma(y)} \ (y \in \dom \Gamma) $$
Then we have $\models \tau : (\sigma, x:=N) = (\sigma, x:=N') : \Gamma , x : A$, and so the induction hypothesis gives
$$ \models M \{ \tau \} : M[\sigma, x:=N] =_B M [ \sigma, x:= N' ] $$
We observe that $M \{ \tau \} \equiv M [ \sigma ] \{ x:=Q:N=N' \}$ (Lemma \ref{lm:pathsubsub}),
and so by Expansion (Lemma \ref{lm:expansion}) and Conversion (Lemma \ref{lm:conv-compute}) we have
$$ \models (\lambda x:A.M[\sigma])\{\}_{N N'} Q : (\lambda x:A.M[\sigma])N =_B (\lambda x:A.M[\sigma])N' $$
as required.
\item
We must show that
$$  \models \triplelambda e : x =_A y. M \{ \tau : \rho = \sigma, x := e : x = y \} : \lambda x:A.M [ \rho ] =_{A \rightarrow B} \lambda x:A.M [ \sigma ] \enspace . $$
So let $ \models P : N =_A N'$.  The induction hypothesis gives
$$  \models M \{ \tau : \rho = \sigma, x := P : N = N' \} : M [\rho, x := N] =_B M [\sigma, x := N'] \enspace , $$
and so we have
\begin{align*}
\models & (\triplelambda e : x =_A y. M \{ \tau : \rho = \sigma, x := e : x = y \})_{N N'} P \\
: & (\lambda x:A.M [ \rho ])N =_B (\lambda x:A.M [ \sigma ])N'
\end{align*}
by Expansion and Conversion, as required.
\end{enumerate}
\end{proof}

\begin{corollary}
Let $\Gamma$ be a context in which no term variables occur.
\begin{enumerate}
\item
If $\Gamma \vdash \delta : \phi$ then $\delta$ reduces to a neutral proof or canonical proof.
\item
If $\Gamma \vdash P : M =_A N$ then $P$ reduces to a neutral path or canonical path.
\end{enumerate}
\end{corollary}

\begin{proof}
Let $\id$ be the substitution $\Gamma \Rightarrow \Gamma$ such that $\id(x) \eqdef x$.
If $\Gamma \vald$ then $\models \id : \Gamma$ using Lemmas \ref{lm:neutral-proof} and \ref{lm:neutral-path}.

Therefore, if $\Gamma \vdash E : T$ then $\models E [ \id ] : T [ \id ]$, that is, $\models E : T$.  Hence $E$ reduces to
a neutral expression or canonical expression.
\end{proof}

\begin{corollary}[Canonicity]
Let $\Gamma$ be a context with no term variables.
\begin{enumerate}
\item
If $\Gamma \vdash \delta : \bot$ then $\delta$ reduces to a neutral proof.
\item
If $\Gamma \vdash \delta : \phi \supset \psi$ then $\delta$ reduces either to a neutral proof, or a proof $\lambda p : \phi' . \epsilon$ where $\phi \simeq \phi'$
and $\Gamma, p : \phi \vdash \epsilon : \psi$.
\item
If $\Gamma \vdash P : \phi =_\Omega \psi$ then $P$ reduces either to a neutral path; or to $\reff{\chi}$ where $\phi \simeq \psi \simeq \chi$; or to
$\univ{\phi'}{\psi'}{\delta}{\epsilon}$ where $\phi \simeq \phi'$, $\psi \simeq \psi'$, $\Gamma \vdash \delta : \phi \supset \psi$ and $\Gamma \vdash \epsilon : \psi \supset \phi$.
\item
If $\Gamma \vdash P : M =_{A \rightarrow B} M'$ then $P$ reduces either to a neutral path; or to $\reff{N}$ where $M \simeq M' \simeq N$; or to
$\triplelambda e : x =_A y. Q$ where $\Gamma, x : A, y : A, e : x =_A y \vdash Q : Mx =_B M'y$.
\end{enumerate}
\end{corollary}

\begin{proof}
A closed expression cannot be neutral, so from the previous corollary every typed closed expression must reduce to a canonical expression.
We now apply case analysis to the possible forms of canonical expression, and use the Generation Lemma.
\end{proof}

\begin{corollary}[Conistency]
There is no $\delta$ such that $\vdash \delta : \bot$.
\end{corollary}

\begin{note}
We have not proved canonicity for terms.  However, we can observe that PHOML restricted to terms and types is just the simply-typed lambda calculus
with one atomic type $\Omega$ and two constants $\bot$ and $\supset$; and our reduction relation restricted to this fragment is head reduction.
Canonicity for this system is already a well-known result (see e.g. \cite[Ch. 4]{Girard1989}).
\end{note}

\section{Conclusion and Future Work}

We have presented a system with propositional extensionality, and shown that it satisfies the property of canonicity.
This gives hope that it will be possible to find a computation rule for homotopy type theory that satisfies canonicity, and
that does not involve extending the type theory, either with a nominal extension of the syntax as in cubical type theory or
otherwise.

We now intend to do the same for stronger and stronger systems, getting
ever closer to full homotopy type theory.  The next steps will be:

\begin{itemize}
\item
a system where the equations $M =_A N$ are objects of $\Omega$, allowing us to form propositions such as $M =_A N \supset N =_A M$.
\item
a system with universal quantification over the types $A$, allowing us to form propositions such as $\forall x:A. x =_A x$ and
$\forall x,y : A. x =_A y \supset y =_A x$
\end{itemize}

Ultimately, we hope to approach full homotopy type theory.  The study of how the reduction relation and its properties change as we
move up and down this hierarchy of systems should reveal facts about computing with univalence that might be lost when working in
a more complex system such as homotopy type theory or cubical type theory.

\bibliography{../../../../type}

\appendix

\section{Calculation in Cubical Type Theory}
\label{appendix:cubical}

\newcommand{\steptwo}{\mathsf{step}_2}
\newcommand{\stepthree}{\mathsf{step}_3}

We can prove that, if $X$ is a proposition, then the type $\Sigma f:\top \rightarrow X. Path \, X \, x \, (f I)$ is contractible (we omit the details).  Let $e[X, x, p]$ be the term such that
\begin{align*} & X : \Prop, x : X.1, p : \Sigma f:\top \rightarrow X.1. \Path{X.1}{x}{(fI)} \\
\vdash & e[X, x, p] : \Path{(\Sigma f:\top \rightarrow X.1. \Path{X.1}{x}{(f I)})}{\langle \lambda t : \top . x, 1_{X.1} \rangle}{p}
\end{align*}

Let $\steptwo [X,x] \eqdef \langle \langle \lambda t : T.x, 1_{X.1} \rangle,
\lambda p : \Sigma f : T \rightarrow X.1. \Path{X.1}{x}{(fI)}. e[X, x, p] \rangle$.  Then
$$ X : \Prop, x : X.1 \vdash \steptwo [X, x] : \mathsf{isContr}(\Sigma f:T \rightarrow X.1. \Path{X.1}{x}{(fI)}) \enspace . $$
Let $\stepthree [X] \equiv \lambda x : X.1. \steptwo [X, x]$.  Then
$$ X : \Prop \vdash \stepthree [X] : \mathsf{isEquiv} \, (T \rightarrow X.1) \, X.1 \, (\lambda f : T \rightarrow X.1. f I) \enspace . $$
Let $E[X] \equiv \langle \lambda f : \top \rightarrow X.1. f I, step3[X] \rangle$.  Then
$$ X : \Prop \vdash E[X] : \mathsf{Equiv} \, (\top \rightarrow X.1) \, X.1 $$

From this equivalence, we want to get a path from $\top \rightarrow X.1$ to $X.1$ in $U$.  We apply the proof of univalence in \cite{cchm:cubical}

Let $P[X] \equiv \langle i \rangle \mathsf{Glue} [(i = 0) \mapsto (\top \rightarrow X.1, E[X]), (i = 1) \mapsto (X.1, equiv^k X.1)] X.1$.  Then
$$ X : \Prop \vdash P[X] : \Path{U}{(\top \rightarrow X.1)}{X.1} $$
Let $Q \equiv \langle i \rangle \lambda x : \Prop. P[X] i$.  Then
$$ \vdash Q : \Path{(\Prop \rightarrow U)}{F}{I} $$

This is the term in cubical type theory that corresponds to $\triplelambda e : x =_\Omega y.P$ in PHOML (formula \ref{eq:llleP}).  We now construct
terms corresponding to formulas (\ref{eq:llleP2}) and (\ref{eq:llleP3}):
$$ \vdash \langle i \rangle H (Q i) : \Path{U}{(\top \rightarrow \top)}{\top} $$
$$ \vdash \lambda x : \top. \comp^i (H (Q (1 - i))) [] x : \top \rightarrow \top \rightarrow \top $$

Let us write $\outputt$ for this term:
$$ \outputt \eqdef \lambda x : \top. \comp^i (H (Q (1 - i))) [] x \enspace . $$
And we calculate (using the notation from \cite{cchm:cubical} section 6.2):
\begin{align*}
& \quad \outputt \\
& = \lambda x : \top. \comp^i (Q (1 - i) \top) [] x \\
& = \lambda x : \top. \comp^i (P[\top] (1 - i)) [] x \\
& = \lambda x : \top. \comp^i (\mathsf{Glue}[(i = 1) \mapsto (\top \rightarrow \top, E[\top]), (i = 0) \mapsto (\top, \mathsf{equiv}^k \top)] \top) [] x \\
& = \lambda x : \top. \mathsf{glue} [ 1_\mathbb{F} \mapsto t_1 ] a_1 \\
& = \lambda x : \top. t_1 \\
& = \lambda x : \top. (\mathsf{equiv} \, E[\top] \, [] \, \mapid{\top}{x}).1 \\
& = \lambda x : \top. (\mathsf{contr} (step2[\top, \mapid{\top}{x}]) []).1 \\
& = \lambda x : \top. (\comp^i \\
& \qquad (\Sigma f : \top \rightarrow \top. \Path{\top}{(\mapid{\top}{x})}{(fI)}) \\
& \qquad [] \\
& \qquad \langle \lambda t : \top. \mapid{\top}{x}, 1_{mapid_\top(x)} \rangle).1 \\
& = \lambda x : \top. \mapid{\top \rightarrow \top}{(\lambda y : \top. \mapid{\top}{x})}
\end{align*}
Therefore,
\begin{align*}
& \outputt \, m \, n \\
& = \mapid{\top \rightarrow \top}{(\lambda y : \top. \mapid{\top}{m})} n \\
& \equiv (\comp^i (\top \rightarrow \top) [] (\lambda _ : \top. \mapid{\top}{m})) n \\
& = \mapid{\top}{\mapid{\top}{m}}
\end{align*}

\section{Proof of Confluence}
\label{section:confluence}

The proof follows the same lines as the proof given in \cite{luo:car}.

\begin{figure}
\paragraph*{Reflexivity}
$$ \infer{E \rhd E}{} $$
\paragraph*{Reduction on Terms}
$$ \infer{(\lambda x:A.M)N \rhd M[x:=N]}{} \quad
\infer{MN \rhd M'N}{M \rhd M'} \quad
\infer{\varphi \supset \psi \rhd \varphi' \supset \psi'}{\varphi \rhd \varphi' \quad \psi \rhd \psi'} $$
\paragraph*{Reduction on Proofs}
$$\infer{(\lambda p : \varphi . \delta)\epsilon \rhd \delta [ p := \epsilon ]}{} \quad
\infer{\reff{\varphi}^+ \rhd \lambda p : \varphi . p}{} \quad
\infer{\reff{\varphi}^- \rhd \lambda p : \varphi . p}{} $$
$$ \infer{\univ{\varphi}{\psi}{\delta}{\epsilon}^+ \rhd \delta}{} \quad
\infer{\univ{\varphi}{\psi}{\delta}{\epsilon}^- \rhd \epsilon}{} $$
$$ \infer{\delta \epsilon \rhd \delta' \epsilon}{\delta \rhd \delta'} \quad \infer{P^+ \rhd Q^+}{P \rhd Q} \quad
\infer{P^- \rhd Q^-}{P \rhd Q}$$
\paragraph*{Reduction on Paths}
$$\infer{(\triplelambda e:x =_A y.P)_{MN} Q \rhd P [x := M, y := N, e := Q]}{} $$
$$ \infer{\reff{\lambda x:A.M}_{N N'} P \rhd M \{ x:=P : N = N' \}}{} $$
$$ \infer{\reff{\varphi} \supset^* \reff{\psi} \rhd \reff{\varphi \supset \psi}}{} $$
$$ \infer{\reff{\varphi} \supset^* \univ{\psi}{\chi}{\delta}{\epsilon} \rhd
\univ{\varphi \supset \psi}{\varphi \supset \chi}{\lambda p:\varphi \supset \psi. \lambda q : \varphi. \delta (pq)}{\lambda p : \varphi \supset \chi. \lambda q : \varphi. \epsilon (pq)}}{} $$
$$ \infer{\univ{\varphi}{\psi}{\delta}{\epsilon} \supset^* \reff{\chi} \rhd
\univ{\varphi \supset \chi}{\psi \supset \chi}{\lambda p : \varphi \supset \chi. \lambda q : \psi. p (\epsilon q)}{\lambda p:\psi \supset \chi. \lambda q : \varphi. p(\delta q)}}{} $$
$$ \infer{\begin{array}{l}
\univ{\varphi}{\psi}{\delta}{\epsilon} \supset^* \univ{\varphi'}{\psi'}{\delta'}{\epsilon'} \\
 \rhd
\univ{\varphi \supset \varphi'}{\psi \supset \psi'}{\lambda p : \varphi \supset \varphi'. \lambda q : \psi. \delta' (p(\epsilon q))}{\lambda p : \psi \supset \psi'. \lambda q : \varphi. \epsilon' (p (\delta q))}
\end{array}}{} $$
$$ \infer{P_{MN} Q \rhd P'_{MN} Q}{P \rhd P'} \quad
\infer{\reff{M}_{NN'}P \rhd \reff{M'}_{NN'}P}{M \rhd N} \quad
\infer{P \supset^* Q \rhd P' \supset^* Q}{P \rhd P' \quad Q \rhd Q'} $$
\caption{Parallel One-Step Reduction}
\label{fig:POSR}
\end{figure}

\begin{definition}[Parallel One-Step Reduction]
Define the notion of \emph{parallel one-step reduction} $\rhd$ by the rules given in Figure \ref{fig:POSR}.
Let $\rhd^*$ be the transitive closure of $\rhd$.
\end{definition}

\begin{lemma}
\label{lm:rhdiff}
\begin{enumerate}
\item
If $E \rightarrow F$ then $E \rhd F$.
\item
If $E \twoheadrightarrow F$ then $E \rhd^* F$.
\item
If $E \rhd^* F$ then $E \twoheadrightarrow F$.
\end{enumerate}
\end{lemma}

\begin{proof}
These are easily proved by induction.
\end{proof}

Our reason for defining $\rhd$ is that it satisfies the diamond property:

\begin{lemma}[Diamond Property]
If $E \rhd F$ and $E \rhd G$ then there exists an expression $H$ such that $F \rhd H$ and $G \rhd H$.
\end{lemma}

\begin{proof}
The proof is by case analysis on $E \rhd F$ and $E \rhd G$.  We give the details for one case here:
$$ \reff{\phi} \supset^* \reff{\psi} \rhd \reff{\phi \supset \psi} \mbox{ and } \reff{\phi} \supset^* \reff{\psi} \rhd \reff{\phi'} \supset^* \reff{\psi'} $$
where $\phi \rhd \phi'$ and $\psi \rhd \psi'$.  In this case, we have $\reff{\phi \supset \psi} \rhd \reff{\phi' \supset \psi'}$ and $\reff{\phi'} \supset^* \reff{\psi'}
\rhd \reff{\phi' \supset \psi'}$.
\end{proof}

\begin{corollary}
If $E \rhd^* F$ and $E \rhd^* G$ then there exists $H$ such that $F \rhd^* H$ and $G \rhd^* H$.
\end{corollary}

\begin{corollary}
If $E \twoheadrightarrow F$ and $E \twoheadrightarrow G$ then $F \twoheadrightarrow H$ and $G \twoheadrightarrow H$.
\end{corollary}

\begin{proof}
Immediate from the previous corollary and Lemma \ref{lm:rhdiff}.
\end{proof}

\end{document}